\documentclass[12pt]{article}

\usepackage{a4wide}
\usepackage{fancyhdr}
\usepackage{graphicx}
\usepackage{german}
\usepackage[bookmarks]{hyperref}
\usepackage[usenames,dvipsnames]{xcolor}
\usepackage[english]{babel} 
\usepackage[color=cyan]{todonotes}
\usepackage{float}
\usepackage{breakcites}
\newcommand{\R}{\mathbb R}
\usepackage{bbm}
\usepackage{mathtools}

\DeclarePairedDelimiter\floor{\lfloor}{\rfloor}
\usepackage{enumitem}
\usepackage[usenames,dvipsnames]{xcolor}
\usepackage{tcolorbox}
\usepackage{tabularx}
\usepackage{array}
\usepackage{colortbl}
\tcbuselibrary{skins}
\usepackage{caption}
\usepackage{subcaption}
\usepackage[margin=1in]{geometry}
\usepackage{tablefootnote}
\usepackage[utf8]{inputenc}
\usepackage[T1]{fontenc}  
\usepackage{hyperref} 
\usepackage{url}
\usepackage[framemethod=tikz]{mdframed}
\usepackage{cuted}
\usepackage{amsmath}
\usepackage{enumitem}
\usepackage{mathrsfs}
\usepackage{relsize}
\usepackage{datetime}
\usepackage{amssymb,amsmath,amsthm,amsfonts}
\usepackage{afterpage}
\usepackage{authblk}
\numberwithin{equation}{section} 

\definecolor{color1}{RGB}{0,139,0} 
\definecolor{color2}{RGB}{154,255,154} 

\newcommand{\N}{\mathbb{N}} 
\newcommand{\ra}{\rightarrow}

\hypersetup{hidelinks,colorlinks,breaklinks=true,urlcolor=color2,citecolor=color1,linkcolor=color1,bookmarksopen=false,pdftitle={Title},pdfauthor={Author}}
\newtheorem{theorem}{Theorem}
\newtheorem{definition}[theorem]{Definition}

\newtheorem{lemma}[theorem]{Lemma}
\newtheorem{remark}[theorem]{Remark}

\newtheorem{fact}[theorem]{Fact}

\newcommand{\PH}{\mathsf{PH}}
\newcommand{\PL}{\mathsf{PL}}
\newcommand{\PP}{\mathsf{PP}}

\newcommand{\X}{\mathsf{X}}
\newcommand{\Y}{\mathsf{Y}}
\newcommand{\E}{\mathsf{E}}
\newcommand{\INF}{\mathsf{INF}}
\newcommand{\UNF}{\mathsf{UNF}}
\title{\textit{Pull} and \textit{Push\&Pull} in Random Evolving Graphs}
\author[1]{Rami Daknama}
\date{\footnotesize{Sunday, 4\textsuperscript{th} February 2018}}
\affil[1]{Ludwig-Maximilians-Universität München, daknama@math.lmu.de}
\begin{document}

\maketitle

\begin{abstract}
The \textit{Push}, the \textit{Pull} and the \textit{Push\&Pull} algorithms are well-known rumor spreading protocols. In all three, in the beginning one node of a graph is informed. In the \textit{Push} setting, every round every informed node chooses a neighbor uniformly at random and, if it is not already informed anyway, informs it. In the \textit{Pull} setting, each round each uninformed node chooses a neighbor uniformly at random and asks it for the rumor; if the asked neighbor is informed, now also the asking node is informed. \textit{Push\&Pull} is a combination of \textit{Push} and \textit{Pull}: In each round, each node picks a neighbor uniformly at random. If at least one of both knows the rumor, after this round, both know the rumor. 
Clementi et al. have considered \textit{Push} in settings where the underlying graph changes each round (\cite{clementi_rumor_2016}). In one setting they investigated, in each round the underlying graph is a newly sampled Erdős-Rényi random graph $G(n,p)$. They show that if $p\geq 1/n$ then with probability $1-o(1)$ (as $n\ra \infty$) the number of rounds needed until all nodes are informed is $\mathcal{O}(\ln(n))$.
Doerr and Kostrygin introduced a general framework to analyze rumor spreading algorithms (\cite{doerr_randomized_2017}); using this framework, for $a>0$ and $p=a/n$ they improved the results from \cite{clementi_rumor_2016} in the described setting. In particular the expected number of rounds needed by \textit{Push} was determined to be $\log_{2-e^{-a}}(n) + 1/(1-e^{-a})\ln(n) + \mathcal{O}(1)$; also large deviation bounds were obtained.  
Using their framework, we investigate \textit{Pull} and \textit{Push\&Pull} in that setting: We prove that the expected number of rounds needed by \textit{Pull} to inform all nodes is $\log_{2-e^{-a}}(n) + 1/a \ln(n) + \mathcal{O}(1)$. Let $\gamma := 2(1-e^{-a})- (1-e^{-a})^2/a$; we prove that the expected number of rounds needed by \textit{Push\&Pull} is $\log_{1+\gamma}(n) + 1/a \ln(n) + \mathcal{O}(1)$; as a byproduct, we obtain large deviation bounds, too.
\end{abstract}

\section{Introduction}
The \textit{Push}, the \textit{Pull} and the \textit{Push\&Pull} algorithms are important and  well-studied rumor spreading protocols \cite{frieze1985,pittel_on_1987,demers_epidemic_1987,feige_randomized_1990, giakkoupis_tight_2011, giakkoupis_tight_2014, panagiotou_randomized_2015, daum_rumor_2016}. In all three, in the beginning one node of a graph is informed. In the \textit{Push} setting, every round every informed node chooses a neighbor uniformly at random and, if it is not already informed anyway, informs it. In the \textit{Pull} setting, each round each uninformed node chooses a neighbor uniformly at random and asks it for the rumor; if the asked neighbor is informed, now also the asking node is informed. \textit{Push\&Pull} is a combination of \textit{Push} and \textit{Pull}: In each round, each node picks a neighbor uniformly at random. If at least one of both knows the rumor, then, after this round, both know the rumor. \\

Recently Clementi et al. have investigated \textit{Push} on random evolving graphs (\cite{clementi_rumor_2016}), i.e.~in a setting where the underlying graph is not fixed but changes over time. One such setting treated in \cite{clementi_rumor_2016} is the following: Each round the underlying graph is a newly (and independently of the previous graphs) sampled Erdős-Rényi random graph $G(n,p)$. We are interested in large values for $n$, thus all asymptotic notation is with respect to $n\ra\infty$ if not explicitly stated differently.  Among other results, in \cite{clementi_rumor_2016} it is shown that if $p\geq 1/n$ then whp (with high probability, i.e.~with probability $1-o(1)$) the number of rounds needed by \textit{Push} is $\mathcal{O}(\ln(n))$.
 Let $a>0$ and let $n>a$ be a natural number. For $p=a/n$ Doerr and Kostrygin have improved this bound (\cite{doerr_randomized_2017}). They have shown that the expected number of rounds needed is $\log_{2-e^{-a}}(n)+ 1/(1-e^{-a})\ln(n) + \mathcal{O}(1)$; moreover, it is shown that constants $\alpha, A >0$ exist such that, if $T_n$ (or short $T$) denotes the needed number of rounds, then for all $r,n \in \N$ we have $P[|T-E[T]|>r]\leq A \exp(-\alpha r)$. This was shown by applying a general framework developed in \cite{doerr_randomized_2017}. This framework exploits that many rumor spreading algorithms are sufficiently characterized by the probability $p_k$ of a node to become informed in a round that starts with $k$ informed nodes and a bound on the covariances between the indicator variables each indicating whether an uninformed node becomes informed in that round.
By bounding $p_k$ and the mentioned covariances, the framework allows to obtain the expected number of rounds needed up to constant additive terms as well as large deviation bounds.\\

We use this framework to investigate \textit{Pull} and \textit{Push\&Pull} in random evolving graphs. 
We show that the expected number of rounds needed by the \textit{Pull} algorithm in the setting described above (i.e.~each round a new $G(n,p)$ is sampled independently of what happened before) is $\log_{2-e^{-a}}(n) + 1/a \ln(n) + \mathcal{O}(1)$. Let $\gamma = 2(1-e^{-a})- (1-e^{-a})^2/a$; then the expected number of rounds needed by \textit{Push\&Pull} is $\log_{1+\gamma}(n) + 1/a \ln(n) + \mathcal{O}(1)$. As a byproduct, we also obtain large deviation bounds. \\

Particularly the results for \textit{Push\&Pull} are interesting. While both, \textit{Push} and \textit{Pull}, need logarithmic time for the last phase of the rumor spreading, when combining them in \textit{Push\&Pull}, \textit{Push} becomes useless in the last phase which might be unexpected. Another interesting aspect is that \textit{Push\&Pull} in the investigated setting is an example where in the first phase when almost no nodes are informed \textit{Push} and \textit{Pull} get in each other's way in the sense that even at the very beginning they inform significantly fewer nodes than the sum of the numbers of nodes they would have informed individually; in other words: even in the beginning many nodes are informed by \textit{Push} as well as by \textit{Pull}.\\

The remainder of this paper is structured as follows: In section \ref{preliminaries} needed preliminaries are considered; in particular this includes the framework introduced in \cite{doerr_randomized_2017}.
In section \ref{pullsection} the result for \textit{Pull} is proven and in section \ref{pushpullsection} the result for \textit{Push\&Pull} is proven.

\section{Preliminaries}  
\label{preliminaries}
We start with stating the framework from \cite{doerr_randomized_2017}.
Therefore we consider only homogeneous rumor spreading processes characterized as follows: We consider graphs with $n$ nodes, in the beginning one node is informed, the other nodes are uninformed. Once a node is informed it remains informed. The process is partitioned into rounds, in each round each uninformed node can become informed. Whenever a round starts with $k$ nodes, we assume that there is a $p_k$ (only depending on $k$) such that each uninformed node becomes informed in that round with probability $p_k$; hence $p_k$ is called the success probability. A rumor spreading process as described is called homogeneous (\cite{doerr_randomized_2017}).
By suitably bounding the success probability and the covariance numbers defined as follows, bounds on the rumor spreading time (see Definition \ref{rumorspreadingtime}) can be obtained.   
 
\begin{definition}[Covariance numbers, \cite{doerr_randomized_2017}]
For a given homogeneous process and $k\in \{1,\dots,n-1\}$ let $c_k$ be the smallest number such that whenever a round starts with $k$ informed nodes for any two uninformed nodes $x_1,x_2$, the indicator random variables $X_1,X_2$ for the events that these nodes become informed in this round satisfy $Cov[X_1,X_2]\leq c_k$.
\end{definition} 

\begin{definition}[Rumor spreading times, \cite{doerr_randomized_2017}]  \label{rumorspreadingtime} Consider a homogeneous rumor spreading process. For all $t=0,1,\dots$ denote $I_t$ the number of informed nodes at the end of the $t$-th round $(I_0:=1)$. Let $k\leq m \leq n$. Let $T_n(k,m)$ (or short $T(k,m)$) denote the time it takes to increase the number of informed nodes from $k$ to $m$ or more, that is, $T(k,m)=\min\{t-s\mid I_s=k \text{ and } I_t \geq m\}$. We call $T(1,n)$ the rumor spreading time of the process.
\end{definition}

If the following exponential growth condition is fulfilled, then Theorem \ref{expgrowth} states that there is an exponential growing phase, i.e.~if few enough nodes are informed, then the number of informed nodes essentially increases by a constant factor each round and the rumor spreading time can be bounded respectively. 

\begin{definition}[Exponential growth conditions, \cite{doerr_randomized_2017}]
Let $\gamma_n$ be bounded between two positive constants. Let $a,b,c\geq 0$ and $0<f<1$. We say that a homogeneous rumor spreading process satisfies the upper (respectively lower) exponential growth conditions in $[1,fn[$ if for any $n\in \N$ big enough the following properties are satisfied for any $k<fn$.
\begin{itemize}
\item $p_k \geq \gamma_n \frac{k}{n}(1-a \frac{k}{n}-\frac{b}{\ln(n)})$ (respectively $p_k \leq \gamma_n \frac{k}{n}(1+a\frac{k}{n}+\frac{b}{\ln(n)})$).
\item $c_k \leq c \frac{k}{n^2}$.
\end{itemize}
In the case of the upper exponential growth condition, we also require $af<1$.
\end{definition}

\begin{theorem}[\cite{doerr_randomized_2017}]
\label{expgrowth}
If a homogeneous rumor spreading process satisfies the upper (lower) exponential growth conditions in $[1,fn[$, then there are constants $A,\alpha>0$ such that 
\begin{align*}
&E[T(1,fn)]\underset{(\geq)}{\leq} \log_{1+\gamma_n}(n) + \mathcal{O}(1),\\
&P[T(1,fn)\underset{(\leq)}{\geq} \log_{1+\gamma_n}(n) \underset{(-)}{+} r] \leq A \exp(-\alpha r) \text{ for all } r, n \in \N.
\end{align*}

When the lower exponential growth conditions are satisfied, then also there is an $f' \in ]f,1[$ such that with probability $1-\mathcal{O}(1/n)$ at most $f'n$ nodes are informed at the end of round $T(1,f n)$.
\end{theorem}

If the following exponential shrinking condition is fulfilled, then Theorem \ref{expshrink} states that there is an exponential shrinking phase, i.e.~if enough nodes are informed, then the number of uninformed nodes essentially decreases by a constant factor each round and the rumor spreading time can be bounded respectively.

\begin{definition}[Exponential shrinking conditions, \cite{doerr_randomized_2017}]
Let $\rho_n$ be bounded between two positive constants. Let $0<g<1$, and $a,c \in \R_{\geq 0}$. We say that a homogeneous rumor spreading process satisfies the upper (respectively lower) exponential shrinking conditions if for any $n \in \N$ big enough, the following properties are satisfied for all $u=n-k\leq g n$.
\begin{itemize}
\item $1-p_k=1-p_{n-u}\leq e^{-\rho_n} + a \frac{u}{n}$ (respectively $1-p_k=1-p_{n-u}\geq e^{-\rho_n}-a\frac{u}{n}$)
\item $c_k=c_{n-u}\leq \frac{c}{u}$
\end{itemize}
For the upper exponential shrinking conditions, we also assume that $e^{-\rho_n}+a g <1$.
\end{definition}

\begin{theorem}[\cite{doerr_randomized_2017}]
\label{expshrink}
If a homogeneous rumor spreading process satisfies the upper (lower) exponential shrinking conditions, then there are $A',\alpha'>0$ such that 
\begin{align*}
&E[T(n-\floor{gn},n)] \underset{(\geq)}{\leq} \frac{1}{\rho_n}\ln(n) + \mathcal{O}(1), \\
&P[T(n-\floor{gn},n)  \underset{(\leq)}{\geq} \frac{1}{\rho_n}\ln(n) \underset{(-)}{+} r] \leq A' \exp(-\alpha' r) \text{ for all } r, n \in \N.
\end{align*}
\end{theorem}

\begin{remark}
\label{additiveChanges}
It suffices to compute $\gamma_n$ and $\rho_n$ from Theorems \ref{expgrowth} and \ref{expshrink} respectively up to additive $\mathcal{O}(1/\ln(n))$ terms. 
\end{remark}

We will use the following two well-known facts in our proofs; Fact \ref{fact_1} is a simple consequence of Fact \ref{fact_2}.

\begin{fact}
\label{fact_2}
Let $(x_n)_{n \in \N}$ be a sequence of real numbers such that for each $n \in \N$ we have $0<x_n<1$. Then $(1+x_n/n)^n=e^{x_n} + \mathcal{O}(x_n^2/n).$
\end{fact}

\begin{fact}
\label{fact_1}
Let $a>0$; consider an Erdős-Rényi random graph $G=G(n,a/n)$. Let $x$ be a node. The probability that $x$ is isolated is $e^{-a} + \mathcal{O}(1/n)$.
\end{fact}

Theorem \ref{push} considers the number of rounds \textit{Push} needs in the described setting. While we do not need the Theorem for the proof of our results, we state it for completeness. 

\begin{theorem}[\cite{doerr_randomized_2017}]
\label{push}
Let $a>0$ and let $T_n$ be the time the push protocol needs to inform all $n$ nodes when in each round a newly sampled Erdős-Rényi random graph $G=G(n,a/n)$ is the underlying graph. Then $$E[T_n]=\log_{2-e^{-a}}(n) + \frac{1}{1-e^{-a}} \ln(n) + \mathcal{O}(1)$$
and there are constants $A,\alpha>0$ such that for all $r,n\in\N$ $$P[|T-E[T]| \geq r]\leq A \exp(-\alpha r).$$
\end{theorem}

It is observed that the obtained rumor spreading time is the same (up to constant terms) as if the underlying graph is a complete graph but message transmissions fail independently with probability $e^{-a}$ which, up to additive $\mathcal{O}(1/n)$ terms is the probability that a vertex is isolated. We will see that this also holds for \textit{Pull}. Interestingly it does not hold for \textit{Push\&Pull}; we provide an explanation in Remark \ref{explanation}.

\section{\textit{Pull} in Random Evolving Graphs}
\label{pullsection}
\begin{theorem}
\label{pull}
Let $a>0$ and assume that each round a newly sampled Erdős-Rényi random graph $G(n,a/n)$ is the underlying graph. Then for the rumor spreading time of \textit{Pull}, $T_n$, we have
 $$E[T_n]=\log_{2-e^{-a}}(n) + \frac{1}{a}\ln(n) + \mathcal{O}(1)$$ 
and there are constants $A,\alpha>0$ such that for all $r,n \in \N$ $$P[|T_n-E[T_n]| \geq r ] \leq A \exp(- \alpha r).$$ 
\end{theorem}

\begin{proof}
We want to apply the framework from \cite{doerr_randomized_2017}. We can assume that at the start of each round, the edges of the random graph are not yet sampled. Before the $G(n,p)$ is sampled, each uninformed node has the same probability of getting informed, hence the rumor spreading algorithm is homogeneous. First we consider the covariance numbers.
To do this, consider two uninformed nodes $x$ and $y$ and let $X$ and $Y$ denote the random indicator variables indicating whether $x$ or $y$ respectively get informed in this round. 
Note that, as the edges are not yet sampled, there is some positive correlation between $X$ and $Y$, because if we condition on the event that the uninformed node $x$ becomes informed, then it is slightly less likely that $x$ and the uninformed node $y$ are neighbors which increases the probability that $y$ has a higher fraction of informed neighbors and therefore $y$ pulls the information more likely. However, the framework from \cite{doerr_randomized_2017} allows for some positive correlation. We will bound the covariance accordingly. 
Let $\X:=``X=1"$, $\Y:=``Y=1"$ and let $E(G)$ denote the edge set of the random graph for the current round; define $\E :=``\{x,y\} \in E(G)"$. 

\begin{align*}
Cov(X,Y) = P[\X \cap \Y] - P[\X]P[\Y]= P[\X]P[\Y \mid \X] - P[\X]P[\Y]= P[\X](P[\Y \mid \X] - P[\Y]).
\end{align*}

Now consider $P[\Y \mid \X]$. We have

\begin{align*}
P[\Y \mid \X] \leq P[\Y \mid \neg \E]= \frac{P[\Y \cap \neg \E]}{P[\neg \E]} \leq \frac{P[\Y]}{P[\neg E]}=\frac{P[\Y]}{1-a/n}=P[Y]+\mathcal{O}(1/n)
\end{align*}

Hence we obtain

\begin{align*}
P[\X](P[\Y \mid \X] - P[\Y]) \leq P[\X] \mathcal{O}(1/n) \leq \frac{k}{n} \mathcal{O}(1/n).
\end{align*}

Therefore the covariance conditions are fulfilled for the exponential growing and shrinking phases.
 
Now we have to estimate the probability $p_k$ for an uninformed node to become informed in a round starting with $k$ informed nodes. 
If an uninformed node has a neighbor, i.e.~if it is not isolated, then with probability $k/(n-1)$ it becomes informed. However, if it is isolated, which according to Fact \ref{fact_1} is the case with probability $e^{-a} + \mathcal{O}(1/n)$, the node does not become informed in this round deterministically. Thus $p_k=(1-e^{-a} + \mathcal{O}(1/n))k/n$. Hence both, upper and lower, exponential growth conditions are fulfilled for arbitrary $0<f<1$ with $\gamma_n=1-e^{-a} + \mathcal{O}(1/n)$. Recall that according to Remark \ref{additiveChanges}, the $\mathcal{O}(1/n)$ term is negligible.

Theorem \ref{expgrowth} therefore yields   
\begin{align*}
E[T_n(1,fn)] = \log_{2-e^{-a}}(n)  + \mathcal{O}(1)
\end{align*}
and that there are $A_1,\alpha_1>0$ such that for all $r,n\in \N$
\begin{align*}
P[T_n(1,fn)\underset{\leq}{\geq} \log_{2-e^{-a}}(n) \underset{-}{+} r] \leq A_1 \exp(-\alpha_1 r).
\end{align*}

Next, for the exponential shrinking conditions, we consider $1-p_{n-u}$. We have 
\begin{align*}
1-p_{n-u}= 1- \frac{n-u}{n-1} (1-e^{-a} + \mathcal{O}(1/n))= e^{-a} + (1-e^{-a})\frac{u}{n} + \mathcal{O}(1/n).
\end{align*} 

The upper and lower exponential shrinking conditions are fulfilled with $\rho_n=a+\mathcal{O}(1/n)$ (because $e^{-a} + \mathcal{O}(1/n) = e^{-a + \mathcal{O}(1/n)}$) for an arbitrary $0<g<1$.
Note that according to Remark \ref{additiveChanges}, the term $\mathcal{O}(1/n)$ is negligible.
Theorem \ref{expshrink} therefore yields  
\begin{align*}
E[T_n(n-\floor{gn},n)] = \frac{1}{a}\ln(n) + \mathcal{O}(1)
\end{align*}
and that there are $A_2, \alpha_2>0$ such that for all $r, n \in \N$
\begin{align*}
P[T_n(n-\floor{gn},n)  \underset{\leq}{\geq} \frac{1}{a}\ln(n) \underset{-}{+} r] \leq A_2 \exp(-\alpha_2 r).
\end{align*}

Thus, considering the exponential growth phase and the exponential shrinking phase together, we obtain the claim.
\end{proof}

\section{\textit{Push\&Pull} in Random Evolving Graphs}
\label{pushpullsection}
\begin{theorem}
\label{pushpull}
Let $a>0$ and let $\gamma:=2(1-e^{-a})-(1-e^{-a})^2/a$. Assume that each round a newly sampled Erdős-Rényi random graph $G(n,a/n)$ is the underlying graph. Then for the rumor spreading time of \textit{Push\&Pull}, $T_n$, we have
 $$E[T_n]=\log_{1+\gamma}(n) + \frac{1}{a}\ln(n) + \mathcal{O}(1)$$  
and there are constants $A,\alpha>0$ such that for all $r,n \in \N$ $$P[|T_n-E[T_n]| \geq r ] \leq A \exp(- \alpha r).$$ 
\end{theorem}

Before we prove Theorem \ref{pushpull} we introduce some notation. Consider an uninformed node $y$ at the beginning of a round that starts with $k\in \N$ informed nodes, let $\mu:=k/n$; we will refer to this round as the current round. Let $\PH_y$ denote the event that $y$ is pushed by an informed node in the current round. Analogously let $\PL_y$ denote the event that $y$ pulls the rumour in the current round from an informed node. Further set $\PP_y=\PH_y \cup \PL_y$, i.e.~$\PP_y$ denotes the event that $y$ is pushed or pulls the rumour in the current round.  For $j\in \{0,1,2,\dots,k\}$ let $\INF_y(j)$ denote the event that $y$ has exactly $j$ informed neighbours $x_1,\dots,x_j$. When we write $\INF_y(j)$ this implicitly defines $x_1,\dots,x_j$. Let $x$ be an informed node; let $\PH_y(x)$ denote the event that $y$ is pushed by $x$ in the current round. Similarly, let $\PL_y(x)$ denote the event that $y$ pulls the information from $x$ in the current round. When an index is clear from the context, it may be omitted.
Asymptotic notation is with respect to $n\ra \infty$ or $\mu \ra 0$ respectively.
We will use Lemma \ref{intersection_venn_diagram} to prove Theorem \ref{pushpull}; it quantifies the probability that an uninformed node pulls the information in the current round conditioned on that it gets also pushed by an informed node.
\begin{lemma}
\label{intersection_venn_diagram}
Let $a>0$ and assume that each round a newly sampled Erdős-Rényi random graph $G(n,a/n)$ is the underlying graph. 
Consider a round that starts with $k\in \N$ informed nodes and set $\mu:=k/n$; assume that the edges are not yet sampled. Let $y$ be an uninformed node. Then 
\begin{align*}
P[\PL_y\mid\PH_y]= \frac{1-e^{-a}}{a} + \mathcal{O}(\mu) \text{ for } \mu \ra 0.
\end{align*}
\end{lemma}  
In order to prove Lemma \ref{intersection_venn_diagram} we will use Lemma \ref{explicit_form_for_sum_intersection_venn_diagram} that provides a closed form for a certain sum.
\begin{lemma}
\label{explicit_form_for_sum_intersection_venn_diagram}
Let $n \in \N$, $\mu \in (0,1)$ and $a \in \R$ with $a> 0$. Then
\begin{align*}
\sum\limits_{i=0}^{(1-\mu)n-1} \binom{(1-\mu)n-1}{i}\left(\frac{a}{n}\right)^i\left(1-\frac{a}{n}\right)^{(1-\mu)n-1-i}\frac{1}{i+1}=\frac{1-(1-\frac{a}{n})^{(1-\mu)n}}{a(1-\mu)}.
\end{align*}
\end{lemma}
\begin{proof} It is
\begin{align*}
& \sum\limits_{i=0}^{(1-\mu)n-1} \binom{(1-\mu)n-1}{i}\left(\frac{a}{n}\right)^i\left(1-\frac{a}{n}\right)^{(1-\mu)n-1-i}\frac{1}{i+1} \\ 
& \qquad \quad = \sum\limits_{i=0}^{(1-\mu)n-1} \frac{((1-\mu)n-1)!}{(i+1)!((1-\mu)n-1-i)!}\left(\frac{a}{n}\right)^i\left(1-\frac{a}{n}\right)^{(1-\mu)n-1-i} \\
& \qquad \quad = \sum\limits_{i=1}^{(1-\mu)n} \frac{((1-\mu)n-1)!}{i!((1-\mu)n-i)!}\left(\frac{a}{n}\right)^{i-1}\left(1-\frac{a}{n}\right)^{(1-\mu)n-i} \\
& \qquad \quad = \frac{n}{a}\left(-\frac{(1-\frac{a}{n})^{(1-\mu)n}}{(1-\mu)n} +\frac{1}{(1-\mu)n} \sum\limits_{i=0}^{(1-\mu)n} \frac{((1-\mu)n)!}{i!((1-\mu)n-i)!}\left(\frac{a}{n}\right)^{i}\left(1-\frac{a}{n}\right)^{(1-\mu)n-i} \right). 
\end{align*}
Let $X\sim Bin((1-\mu) n,a/n)$. It is 
\begin{align*}
1=\sum\limits_{i=0}^{(1-\mu) n }P[X=i]=\sum\limits_{i=0}^{(1-\mu)n} \frac{((1-\mu)n)!}{i!((1-\mu)n-i)!}\left(\frac{a}{n}\right)^{i}\left(1-\frac{a}{n}\right)^{(1-\mu)n-i}.
\end{align*} 
Hence we arrive at 
\begin{align*}
\sum\limits_{i=0}^{(1-\mu)n-1} \binom{(1-\mu)n-1}{i}\left(\frac{a}{n}\right)^i\left(1-\frac{a}{n}\right)^{(1-\mu)n-1-i}\frac{1}{i+1} &= \frac{n}{a}\left(-\frac{(1-\frac{a}{n})^{(1-\mu)n}}{(1-\mu)n} + \frac{1}{(1-\mu)n} \right)\\
&= \frac{1-(1-\frac{a}{n})^{(1-\mu)n}}{a(1-\mu)}. \qedhere
\end{align*}
\end{proof}
\begin{proof}[Proof of Lemma \ref{intersection_venn_diagram}]
We will omit $y$ as an index in this proof, i.e.~we will write $\PH$ instead of $\PH_y$ and so on.
First we verify that for all $j\in \{0,1,\dots,\mu n\}$
\begin{align}
\label{eq_1_intersection_venn_diagram_proof}
P[\PH\mid \INF(j)] \leq j P[\PH \mid \INF(1)]. 
\end{align}
It is 
\begin{align*}
P[\PH \mid \INF(j)] = P[\PH(x_1)\cup \dots \cup \PH(x_j)\mid \INF(j)]. 
\end{align*}
Hence by applying the union bound
\begin{align*}
P[\PH \mid \INF(j)] \leq j P[\PH(x_1)\mid \INF(j)]=j P[\PH(x_1)\mid \INF(1)]=j P[\PH \mid \INF(1)]
\end{align*}
which implies (\ref{eq_1_intersection_venn_diagram_proof}).
Similarly we verify that for all $j\in \{1,\dots,\mu n\}$
\begin{align}
\label{eq_2_intersection_venn_diagram_proof}
P[\PH\mid \INF(j)] \geq P[\PH \mid \INF(1)].
\end{align}
It is 
\begin{align*}
P[\PH \mid \INF(j)] &= P[\PH(x_1)\cup (\PH(x_2)\cup \dots \cup \PH(x_j))\mid \INF(j)] \geq P[\PH(x_1)\mid \INF(j)]\\&=P[\PH(x_1)\mid \INF(1)]=P[\PH\mid \INF(1)]
\end{align*}
which implies (\ref{eq_2_intersection_venn_diagram_proof}).
Next we verify that for all $j \in \{1,2,\dots,\mu n\}$
\begin{align}
\label{eq_3_intersection_venn_diagram_proof}
\frac{P[\INF(j)\mid \PH]}{P[\INF(1) \mid \PH]}\leq j \frac{P[\INF(j)]}{P[\INF(1)]}.
\end{align}
Using Bayes' Theorem and (\ref{eq_1_intersection_venn_diagram_proof}) we obtain
\begin{align*}
P[\INF(j)\mid \PH]= \frac{P[\PH\mid \INF(j)]P[\INF(j)]}{P[\PH]} \leq \frac{j P[\PH\mid \INF(1)] P[\INF(j)]}{P[\PH]}.
\end{align*} 
Hence, by again applying Bayes' Theorem we arrive at
\begin{align*}
P[\INF(j)\mid \PH] \leq j \frac{P[\INF(1) \mid \PH]P[\INF(j)]}{P[\INF(1)]}.
\end{align*}
This implies (\ref{eq_3_intersection_venn_diagram_proof}).
Analogously (using (\ref{eq_2_intersection_venn_diagram_proof}) instead of (\ref{eq_1_intersection_venn_diagram_proof})) one verifies
\begin{align}
\label{eq_4_intersection_venn_diagram_proof}
\frac{P[\INF(j)\mid \PH]}{P[\INF(1) \mid \PH]}\geq \frac{P[\INF(j)]}{P[\INF(1)]}.
\end{align}
Next we show that for any $j\in \{1,2,\dots,\mu n\}$ 
\begin{align}
\label{eq_5_intersection_venn_diagram_proof}
\frac{P[\INF(j)]}{P[\INF(1)]}=\mathcal{O}(\mu^{j-1}) \text{ for } \mu \ra 0.
\end{align}
We have 
\begin{align*}
\frac{P[\INF(j)]}{P[\INF(1)]} &= \frac{\binom{\mu n}{j} (\frac{a}{n})^j (1-\frac{a}{n})^{\mu n -j}}{\binom{\mu n}{1}\frac{a}{n}(1-\frac{a}{n})^{\mu n-1}} =
\frac{(\mu n)!}{j! (\mu n - j)! \mu n}\left(\frac{a}{n}\right)^{j-1}\left(1-\frac{a}{n}\right)^{-j+1} \\ &= \frac{\mu n-1}{n}\frac{\mu n -2}{n}\cdot \dots \cdot \frac{\mu n - j +1}{n} \cdot \frac{a^{j-1}}{j!}\left(1-\frac{a}{n}\right)^{-j+1} = \mathcal{O}(\mu^{j-1})
\end{align*}
which shows (\ref{eq_5_intersection_venn_diagram_proof}).
Using (\ref{eq_3_intersection_venn_diagram_proof}), (\ref{eq_4_intersection_venn_diagram_proof}) and (\ref{eq_5_intersection_venn_diagram_proof}) we can infer that for all $j\in \{1,2,\dots,\mu n\}$
\begin{align}
\label{eq_6_intersection_venn_diagram_proof}
\frac{P[\INF(j)\mid \PH]}{P[\INF(1)\mid \PH]}=\mathcal{O}(\mu^{j-1}).
\end{align}
Now we prove
\begin{align}
\label{eq_7_intersection_venn_diagram_proof}
P[\INF(1)\mid \PH]=1+\mathcal{O}(\mu).
\end{align} 
Using (\ref{eq_6_intersection_venn_diagram_proof}) we get 
\begin{align*}
\frac{P[\INF(2)\mid \PH]+P[\INF(3)\mid \PH] + \dots + P[\INF(\mu n)\mid \PH]}{P[\INF(1) \mid \PH]} = \mathcal{O}(\mu).
\end{align*}
Therefore 
\begin{align*}
P[\INF(2)\mid \PH]+P[\INF(3)\mid \PH] + \dots + P[\INF(\mu n)\mid \PH]=P[\INF(1) \mid \PH]\cdot \mathcal{O}(\mu)
\end{align*}
and thus
\begin{align*}
\underbrace{P[\INF(1)\mid \PH]+P[\INF(2)\mid \PH]+ \dots + P[\INF(\mu n)\mid \PH]}_{=1}=P[\INF(1) \mid \PH]\cdot (1+\mathcal{O}(\mu)).
\end{align*}
Hence 
\begin{align*}
(1+\mathcal{O}(\mu))P[\INF(1)\mid \PH]=1
\end{align*}
which implies (\ref{eq_7_intersection_venn_diagram_proof}).
Using (\ref{eq_7_intersection_venn_diagram_proof}) we obtain 
\begin{align}
P[\PL\mid \PH] = P[\PL\mid \INF(1)] + \mathcal{O}(\mu).
\end{align} 
Thus, to finish the proof, it suffices to show
\begin{align}
\label{eq_8_intersection_venn_diagram_proof}
P[\PL\mid \INF(1)]=\frac{1-e^{-a}}{a} + \mathcal{O}(\mu).
\end{align}
For each $j\in \{0,1,\dots,(1-\mu)n-1\}$ let $\UNF(j)$ denote the event that $y$ has exactly $j$ uninformed neighbours in the current round. 
Note that at the beginning of the round there is a fixed number of informed nodes, namely $\mu n$, and a fixed number of uninformed nodes, namely $(1-\mu)n$. In particular, for any $j \in \{0,1,\dots,(1-\mu)n-1\}$, $\UNF(j)$ and $\INF(1)$ are independent.
Hence we have 
\begin{align*}
P[\PL \mid \INF(1)] &= \sum\limits_{i=0}^{(1-\mu)n-1}P[\UNF(i)]\frac{1}{i+1}\\ &=\sum\limits_{i=0}^{(1-\mu)n-1} \binom{(1-\mu)n-1}{i}\left(\frac{a}{n}\right)^i\left(1-\frac{a}{n}\right)^{(1-\mu)n-1-i}\frac{1}{i+1}.
\end{align*} 
Thus, using Lemma \ref{explicit_form_for_sum_intersection_venn_diagram}, we can infer
\begin{align*}
P[\PL\mid \INF(1)]= \frac{1-(1-\frac{a}{n})^{(1-\mu)n}}{a(1-\mu)}.
\end{align*}
Using Fact \ref{fact_2}, this gives
\begin{align*}
P[\PL\mid \INF(1)] = \frac{1-e^{a(\mu - 1)}}{(1-\mu)a} + \mathcal{O}(1/n)=\frac{1-e^{a(\mu - 1)}}{(1-\mu)a} + \mathcal{O}(\mu).
\end{align*}
Thus, using the series representation of the exponential function at zero, we obtain
\begin{align*}
P[\PL\mid \INF(1)] = \frac{1-e^{-a}}{a} + \mathcal{O}(\mu)
\end{align*}
which shows (\ref{eq_8_intersection_venn_diagram_proof}) and hence completes the proof.
\end{proof}

\begin{proof}[Proof of Theorem \ref{pushpull}]
We want to use the framework from \cite{doerr_randomized_2017}.
To do this, consider a round of the rumour spreading process that starts with $k$ informed and $u=n-k$ uninformed nodes; we will refer to this round as the current round. Let $\mu:=k/n$. We can assume that at the start of the round, the edges of the random graph are not yet sampled.
We start with showing that the covariance conditions are fulfilled. Therefore consider two uninformed nodes $x$ and $y$. As before, $\E$ denotes the event that $x$ and $y$ become neighbours in the current round.
We have
\begin{align*}
Cov(\mathbbm{1}_{\PP_x},\mathbbm{1}_{\PP_y})=P[\PP_x \cap \PP_y] - P[\PP_x]P[\PP_y]=P[\PP_x](P[\PP_y\mid \PP_x]-P[\PP_y]).
\end{align*}
It is 
\begin{align*}
P[\PP_y\mid \PP_x] &\leq P[\PP_y\mid \neg E] = \frac{P[\PP_y \cap \neg E]}{P[\neg E]} \leq \frac{P[\PP_y]}{P[\neg E]}= \frac{P[\PP_y]}{1-a/n}=P[\PP_y] + \mathcal{O}(1/n).
\end{align*}
Hence 
\begin{align}
\label{cov_calculation}
Cov(\mathbbm{1}_{\PP_x},\mathbbm{1}_{\PP_y})=P[\PP_x]\cdot \mathcal{O}(1/n).
\end{align}
From \cite{doerr_randomized_2017} it is known that $$P[\PH_x]\leq k/n(1-e^{-a}+\mathcal{O}(1/n))$$
and therefore
$$P[\PP_x]\leq P[\PH_x]+P[\PL_x]\leq (1-e^{-a} + \mathcal{O}(1/n))\frac{k}{n} + (1+\mathcal{O}(1/n))\frac{k}{n} = (2-e^{-a}+\mathcal{O}(1/n))\frac{k}{n}.$$
This, together with (\ref{cov_calculation}), yields
\begin{align*}
Cov(\mathbbm{1}_{\PP_x},\mathbbm{1}_{\PP_y}) \leq (2-e^{-a})k/n \cdot \mathcal{O}(1/n).
\end{align*}
Hence the covariance conditions are fulfilled for the exponential growth and shrinking conditions.

For the exponential growth phase, we have to estimate the success probability $p_k=P[\PP_y]$ that an uninformed node $y$ becomes informed in the current round that starts with $k$ informed nodes. 
In the following, we write $\PP, \PH$ and $\PL$ instead of $\PP_y, \PH_y$ and $\PL_y$ respectively.
Note that $\PH$ and $\PL$ are not independent (as the edges are not sampled yet at the beginning of the round).

It is
\begin{align}
\label{venn_diagram_calculation}
P[\PP]=P[\PH \cup \PL]=P[\PH] + P[\PL] - P[\PH \cap \PL]. 
\end{align}
To compute $P[\PP]$ we consider the three summands of (\ref{venn_diagram_calculation}) individually:

\textbf{Term 1} $P[\PH]$: From \cite{doerr_randomized_2017} it is known that 
\begin{align}
\label{term1_summand}
\mu (1-e^{-a})\left(1-\frac{k + \mathcal{O}(1)}{2n}(1-e^{-a})\right) \leq P[\PH] \leq \mu(1-e^{-a} + \mathcal{O}(1/n)).
\end{align}
\textbf{Term 2} $P[\PL]$: According to Fact \ref{fact_1}, $y$ is isolated with probability $e^{-a}+ \mathcal{O}(1/n)$. Thus we have 
\begin{align}
\label{term2_summand}
P[\PL]=(1-e^{-a}+ \mathcal{O}(1/n))\mu.
\end{align}
\textbf{Term 3} $P[\PL\cap \PH]$: We have  
\begin{align}
\label{term3_summand}
P[\PL\cap \PH]=P[\PL\mid \PH] P[\PH].
\end{align}
Thus, using Lemma \ref{intersection_venn_diagram} and (\ref{term1_summand}) we obtain 
\begin{align*}
P[\PL\cap\PH]=\mu \frac{(1-e^{-a})^2}{a} + \mathcal{O}(\mu^2) \text{ for } \mu \ra 0 .
\end{align*}
Combining the three terms in (\ref{venn_diagram_calculation}), where asymptotic notation is with respect to $\mu \ra 0$, we obtain 
\begin{align*}
P[\PP]=\left(2(1-e^{-a})-\frac{(1-e^{-a})^2}{a} + \mathcal{O}(\mu)\right)\mu= \left(2(1-e^{-a})-\frac{(1-e^{-a})^2}{a}\right)(1+\mathcal{O}(\mu))\mu.
\end{align*} 
In particular there is an $a^* \geq 0 $ such that 
$$p_k=P[\PP] \underset{\leq}\geq \left (2(1-e^{-a}) - \frac{(1-e^{-a})^2}{a} \right)\frac{k}{n}(1 \underset{+}- a^* \frac{k}{n}).$$ 
Hence there is a constant $f>0$ such that the exponential growth conditions are fulfilled for $\gamma=2(1-e^{-a}) - (1-e^{-a})^2/a$.
Thus Theorem \ref{expgrowth} yields 
\begin{align*}
&E[X_n(1,fn)]= \log_{1+\gamma}(n)  + \mathcal{O}(1)
\end{align*}
and that there are constants $A_1,\alpha_1>0$ such that for all $r, n \in \N$
\begin{align*}
&P[X_n(1,fn)\underset{\leq}{\geq} \log_{1+\gamma}(n) \underset{-}{+} r] \leq A_1 \exp(-\alpha_1 r).
\end{align*}
 
Let $g \in (0,1)$ be an arbitrary constant. To complete the proof we show that $1/a \ln(n) + \mathcal{O}(1)$ is a lower bound for the number of rounds needed to inform all remaining nodes, starting with $\floor*{gn}$ informed nodes;
then the claim follows as \textit{Pull} provides a matching upper bound for the exponential shrinking phase.
Consider an uninformed node $y$. According to Fact \ref{fact_1}, $y$ is isolated in the current round with probability $e^{-a} + \mathcal{O}(1/n)=e^{-a+\mathcal{O}(1/n)}$. If $y$ is isolated, then it cannot be informed in the current round. Therefore 
$$1-P[\PP]\geq e^{-a} + \mathcal{O}(1/n).$$
Thus the lower exponential shrinking conditions are fulfilled for $\rho_n= a + \mathcal{O}(1/n)$ and $g\in (0,1)$ can indeed be chosen arbitrarily. 
 Therefore Theorem \ref{expshrink} yields  
\begin{align*}
&E[X_n(n-\floor*{gn},n)] \geq \frac{1}{a}\ln(n) + \mathcal{O}(1)
\end{align*}
and that there are $A_2, \alpha_2>0$ such that for all $r, n \in \N$
\begin{align*}
&P[X_n(n-\floor*{gn},n)  \leq \frac{1}{a}\ln(n) - r] \leq A_2 \exp(-\alpha_2 r).
\end{align*}
Together with the upper bounds that we obtain by considering \textit{Pull}, this completes the proof.
\end{proof}

\begin{remark}
\label{explanation}
It is interesting that \textit{Push\&Pull} does --- unlike \textit{Push} and \textit{Pull} --- behave differently on random evolving graphs than on the complete graph with message transmission success probability $1-e^{-a}$ (which is the probability that a node is not isolated (up to an additive $\mathcal{O}(1/n)$ term)).
The reason for this is that push and pull operations get in each other's way, i.e.~it has a relevant impact that some nodes get informed in the same round by a push as well as by a pull which makes one of those operations useless. This is in contrast to the situation in the complete graph with message transmission success probability $1-e^{-a}$: There in the beginning, \textit{Push} and \textit{Pull} essentially do not get in each other's way, i.e.~only very few nodes get informed by \textit{Push} as well as by \textit{Pull} in the beginning of the rumor spreading process. The reason for this difference is that in the complete graph each node has $n-1$ neighbors while in the setting of this paper the expected number of neighbors of a node is in each round $a + \mathcal{O}(1/n)$ and therefore here it is much more likely that a relevant fraction of edges is used by \textit{Push} as well as by \textit{Pull}.

Another interesting aspect is the behavior of \textit{Push\&Pull} in the last phase of the process:
As in the investigated setting both, \textit{Push} and \textit{Pull}, need logarithmic time for the exponential shrinking phase, one might conjecture that in the \textit{Push\&Pull} setting \textit{Push} as well as \textit{Pull} contribute substantially to the last phase. The reason why this is not the case, i.e.~why only \textit{Pull} contributes to the last phase, is the following:
Consider an uninformed node $x$ in the last phase, i.e.~if most nodes are informed already. In each round we first sample whether $x$ is isolated which, according to Fact \ref{fact_1}, is the case with probability $e^{-a} + \mathcal{O}(1/n)$. If $x$ is isolated it cannot be informed in that round, neither by a push nor by a pull. However, if $x$ is not isolated it is extremely probable that it becomes informed by a pull attempt. In particular the case that it does become informed by a push but not simultaneously also by a pull is very unlikely. So the problem essentially is that both, pull and push attempts, have to clear the same hurdle, i.e.~wait for a round where $x$ is not isolated. But after taking this hurdle it is extremely unlikely that a pull does not succeed while a push attempt does succeed.
\end{remark}

\newpage
\bibliography{literatur}
\bibliographystyle{abbrv}
\end{document}